%% file: main.tex
\documentclass{article}
\usepackage{graphicx} 
\usepackage{prem}
\usepackage{tikzit}
\input{figures/angle.tikzstyles}
\usepackage{fullpage}

\title{Quantum Advantage with Faulty Oracle}

\author[$\dagger$]{David Rasmussen Lolck}
\author[$\dagger$]{Laura Man\v{c}inska}
\author[$\dagger$]{Manaswi Paraashar}
\affil[$\dagger$]{University of Copenhagen, Denmark}

\date{}

\begin{document}

\maketitle

\begin{abstract}
    This paper investigates the impact of noise in the quantum query model, a fundamental framework for quantum algorithms. We focus on the scenario where the oracle is subject to non-unitary (or irreversible) noise, specifically under the \textit{faulty oracle} model, where the oracle fails with a constant probability and acts as identity. Regev and Schiff (ICALP'08) showed that quantum advantage is lost for the search problem under this noise model. Our main result shows that every quantum query algorithm can be made robust in this noise model with a roughly quadratic blow-up in query complexity, thereby preserving quantum speedup for all problems where the quantum advantage is super-cubic. This is the first non-trivial robustification of quantum query algorithms against an oracle that is noisy.
\end{abstract}

\input{introduction}
\input{technical-overview}
\input{preliminaries}
\input{faulty-oracle-noise}

\input{full-oracle-replacement}
\input{extending-to-multiple-bits}

\input{arbitrary-algorithms}
\bibliography{ref}
\end{document}

%% file: figures/angle.tikzstyles
\tikzstyle{arrow}=[->, draw=blue]

%% file: introduction.tex
\section{Introduction}

Understanding and mitigating noise is one of the most important challenges in quantum computing today and the quantum query model is one of the most natural models of computation to study the effects of noise.
In this model, the goal is to solve a problem $P$ on an unknown input $f \in \{0,1\}^n \to \{0,1\}^m$. The input $f$ can only be accessed by a unitary matrix $O_f$, called \textit{oracle}, that satisfies $O_f \ket{z}\ket{b} = \ket{z} \ket{b \oplus f(z)}$ for all $z \in \{0,1\}^n$ and $b \in \{0,1\}^m$, where $z \oplus f(z)$ denotes the bitwise XOR of $b$ and $f(z)$. 
Furthermore, the goal is to do so with bounded-error probability\footnote{By bounded-error we mean correctly with probability at least $4/5$. Algorithms that have success probability $(1/2 + \delta)$ can be made bounded-error by amplitude amplification~\cite{brassard2002quantum} with an overhead of $O(1/\delta)$.}. by making minimum number queries to $O_f$. 
Most of the provable advantages of quantum algorithms over classical ones are in this model (for example, Grover's algorithm for the search problem (\cite{grover1996fast}) and Shor's Algorithm for the period finding problem (\cite{shor1999polynomial})) and this model is naturally connected to many other areas of quantum computing due to its simplicity. Thus, understanding the effects of noise in the query model is a fundamental question.

In this paper, we are interested in the case when the oracle is affected by noise\footnote{There has been work dealing with noise affecting other components of query algorithms (see~\cite{R03} and the references therein).}.
There are two types of noise that may affect the oracle: unitary (or reversible) and non-unitary (or irreversible). An example of a reversible noisy oracle is a low-depth quantum circuit implemented on a fault-tolerant quantum quantum computer. However, when such a circuit is implemented on a quantum computer that is not fault-tolerant and is affected by noise due to its interaction with the environment, we have an oracle with non-reversible noise.  The case of reversible noise is well-understood and easy to mitigate: a quantum algorithm designed for a fault-free oracle with cost $q$ can be easily modified to give an algorithm for a with cost $O(q \ln q)$ that works correctly on an oracle with reversible noise (\cite{hoyer2003quantum, suzuki2006robust, buhrman2007robust, iwama2006query}).

Understanding irreversible noise in the query model has proved to be much more challenging.
The majority of existing results are limited to the study of the effects 
irreversible noise for Grover's algorithm~\cite{PR99, LLZT00, SBW03, FG98, ABNR12, KNR18} or for the search problem itself~\cite{RS08, VRRW14, R03}.
To the best of our knowledge, search is the only problem for which lower bounds under a noisy oracle are known. The work~\cite{RS08} gave the first such lower bound under the following noise model: we are given access to a faulty oracle $O_{f,p}$ which acts as follows
\begin{align*}
    O_{f,p}(\cdot) = \begin{cases}
    O_f(\cdot) & \text{ with probability $(1-p)$}\\
    I(\cdot) & \text{ with probability $p$},
\end{cases} \tag*{(Faulty Oracle) \qquad}
\end{align*}
where $I$ is the identity map, i.e., $I(\rho) = \rho$ for all matrices $\rho$.
Similar to~\cite{RS08}, we refer to this error model as the \textit{faulty oracle} model.
The parameter $p$ is called the error rate. The faulty oracle model, when the error rate $p$ is at most  $1/2$, is the central object of study in this paper. 
It is natural to consider the regime where the error-rate is a constant independent of the system size. This is what we focus on.


The work \cite{RS08} showed an $\Omega(p2^n)$ lower bound for the search problem in the faulty oracle model. This means that quantum advantage is lost when the error rate is a constant. This lower bound has impacted several areas of research in quantum computing. Quantum random access memory (qRAM) that stores classical information, i.e.~bits, was considered in~\cite{AGJMP15}. Such a qRAM can be used as an oracle for quantum query algorithms. They defined an error model that is quite similar to the faulty oracle model. Using the lower bound of~\cite{RS08}, it was shown in~\cite{AGJMP15} that a qRAM that produces queries with constant error will not lead to a quantum advantage for the search problem. Furthermore,~\cite{AGJMP15} conjectured that their error model nullifies the asymptotic speed-ups of other quantum query algorithms as well.
Most supervised quantum machine learning algorithms assume the existence of qRAM (\cite{GLM08}). In view of this and~\cite{RS08}, the authors of~\cite{wossnig2021quantum} suggested that even small errors can crucially impact quantum machine learning algorithms.
Query complexity of the search problem has been studied under noise models other than the faulty oracle model:~\cite{R03} gave an $\Omega(p 2^n)$ lower bound for the search problem under the depolarizing noise model.

A natural question to ask is whether there are problems for which a quantum advantage is retained even under a noisy oracle.

\begin{center}
    \textit{\textbf{Question.} Is there any reasonable fault model for which a quantum speed-up is achievable? (\cite{RS08}) Is a significant quantum speedup is ever possible with a faulty oracle? (\cite{childs2023streaming})}
\end{center}

If we restrict our attention to quantum algorithms where the (parallel) query depth is upper bounded by a constant, for example Simon's algorithm~\cite{simon1997power}, then the answer to the question is indeed yes (with a polynomial overhead) under the faulty oracle model.~\cite{childs2023streaming} studied the case when the oracle has depolarising noise. They showed that quantum algorithms with constant query depth can be made robust with a polynomial overhead. However, their techniques do not extend to algorithms that have $\omega(1)$ query depth without a substantial overhead. Prior to our work, there were no known techniques allowing one to maintain quantum advantage in the noisy regime beyond the case of constant (parallel) query depth.

The main contribution of this paper is to show that in the faulty oracle model, a quantum advantage is retained for \textit{every} quantum query algorithm that is sufficiently fast, in the regime of a constant error rate $p\leq 1/2$.
In the next section, we present our main result which shows how to make a quantum algorithm robust to faulty oracle with a roughly quadratic blow-up in query complexity.

\section{Our Results}

We state our main theorem below.

\begin{theorem}
\label{thm: Informal Main Theorem}
    Consider any $f : \{0,1\}^n \to \{0,1\}^m$ and a bounded-error quantum query algorithm $\mathcal{A}$ that solves a problem $P$ by making $q$ queries to a (fault-free) oracle $O_f$. Then there exists a bounded-error quantum algorithm $\mathcal{A}'$ that solves $P$ and makes $O(q^3 m^2 \ln q)$ queries to the faulty oracle $O_{f,p}$ where $p \leq 1/2$.
\end{theorem}



While quantum advantage is lost for the search problem with a faulty oracle (\cite{RS08}), Theorem~\ref{thm: Informal Main Theorem} implies that a quantum advantage is retained for all problems for which there is a quantum algorithm with roughly super-cubic advantage. Furthermore, no additional assumptions are made about the structure of the algorithms, such as bounded query depth as considered in~\cite{childs2023streaming}, in Theorem~\ref{thm: Informal Main Theorem}. There are several quantum query algorithms that have super-cubic, super-polynomial, and even exponential advantage over their classical counterparts, and from Theorem~\ref{thm: Informal Main Theorem}, a quantum advantage is retained for these algorithms. We refer the reader to~\cite{Zoo} which contains an extensive collection of such quantum algorithms.

The following question was asked in~\cite{RS08}.

\begin{openproblem}[\cite{RS08}]
    Is there any search problem for which a quantum speed-up is achievable with a faulty oracle?
\end{openproblem}

Theorem~\ref{thm: Informal Main Theorem} resolves this problem in the affirmative when combined with the following search algorithms: super-polynomial speedup for the welded tree problem which involves finding a special vertex in a certain graph (\cite{childs2003exponential, li2024recovering, jeffery2023multidimensional, belovs2024global})
and \textit{quartic} speedup for searching counterfeit coins given a bunch of coins (\cite{iwama2012quantum}).

We conclude with some ideas that went into the proof of Theorem~\ref{thm: Informal Main Theorem}.

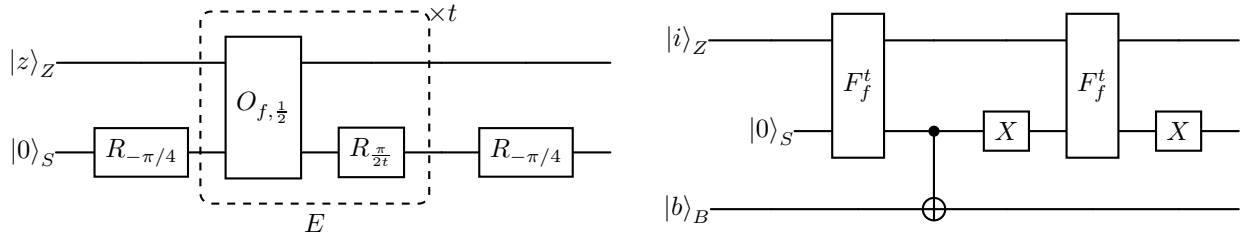
\begin{figure}[t]

\begin{minipage}[c]{0.35\linewidth}
        \begin{quantikz}
        \ket{z}_Z&
        &\gate[2]{O_{f,\frac{1}{2}}}\gategroup[2,steps=2,style={rounded corners,dashed,inner sep=6pt, 
        label = below:$E$
        },label style={label
position=above right,anchor=north,xshift=0.40cm,yshift=0cm}]{$\times t$} &  & & & \\
        \ket{0}_S & \gate{R_{-\pi/4}} &  & \gate{R_{\frac{\pi}{2t}}} &  & \gate{R_{-\pi/4}} &
        \end{quantikz}
\end{minipage} \qquad \qquad \qquad \qquad
\begin{minipage}[c]{0.3\linewidth}
        \begin{quantikz}
        \ket{i}_Z& & \gate[2]{F_f^t} & & & \gate[2]{F_f^t} & & \\
        \setwiretype{n}&\ket{0}_S&\setwiretype{q} & \ctrl{1} & \gate{X} &  & \gate{X} &\\
        \ket{b}_B&&&\targ{}&&&&
        \end{quantikz}
\end{minipage}
\caption{The circuit to the left is $F_f^t$ and the  circuit to right is $G_f^t$. 
}
\label{fig: intro figure}
\end{figure}



\subsection{Ideas behind proof of Theorem~\ref{thm: Informal Main Theorem}}

We present our ideas for the case when $m = 1$, i.e., when $f$ has domain $\{0,1\}^n$ and range $\{0,1\}$. The proof for $m \geq 2$ follows from this case with a careful error-analysis (see Section~\ref{sec:multi-bit-functions}). Let $\cA$ be an algorithm that makes $q$ queries to the fault-free oracle $O_f$ and solves the problem $P$ with bounded-error. Such an algorithm uses three registers:
$Z$, which stores the index of the query; 
$B$, which stores the output of the query; and 
$T$, which is the workspace used by the algorithm. Let 
$$
\cA = U_q (O_f \otimes I_T) U_{q-1} (O_f \otimes I_T) \ldots U_1 (O_f \otimes I_T)U_0
$$
where $U_i$'s act on $Z,B$ and $T$ while $O_f$ acts on $Z$ and $B$. The robust version of $\cA$, denoted $\cA'$, uses an additional register $S$ (called the scratch-pad) and has the form:
$$
\cA' = (U_q\otimes I_S)  (G_f^t\otimes I_T) \ldots (U_1\otimes I_S) (G_f^t\otimes I_T) (U_0 \otimes I_S)
$$
where 
$t = O(q^2 \ln q)$ and 
the unitaries $G_f^t$ are constructed from $O_{f,p}$, as shown in Figure~\ref{fig: intro figure}.
With high probability over the randomness of $O_{f,p}$, $G_f^t$
satisfy the following property: for a suitable choice of $t = O(q^2 \ln q)$ and for all quantum states $\ket{\psi}_{ZBT}$ 
$$
\norm{((O_f \otimes I_{ST})-(G_f^t \otimes I_T))\ket{\psi}_{ZBT}\ket{0}_S} \leq O(1/q).
$$
By applying the triangle inequality and discarding the scratch-pad register at the end of the computation, we can show that the algorithms $\cA$ and $\cA'$ are close to each other. Thus, with a multiplicative overhead of $O(q^2 \ln q)$, we demonstrate how to make a $q$-query bounded-error quantum algorithm robust to a faulty oracle. Note that this error-reduction procedure is independent of the algorithm $\cA$ and works by making each fault-free oracle call robust to noise.



The construction of  $G_f^t$ can be broken down into two steps, the first of which is the construction of the circuit $F_f^t$, shown in Figure~\ref{fig: intro figure}. We show that if the output register $B$ is in the state $\ket{0}$ then, with high probability, $F_f^t$ acts almost identically to $O_f$. Since for all $p$, the faulty oracle $O_{f,p}$ acts as identity on inputs $\ket{z,b}$ where $z\in\{0,1\}^n$ and $f(z) = 0$, this case is easy to analyze. When $f(z) = 1$, the action of this circuit can be seen as a random walk on a line. We show that for $p = 1/2$, $t = O(q^2 \ln q)$ repetitions of $E$ achieve sufficient concentration for the underlying random walk. This completes the analysis for the case where the $B$-register is in the state $\ket{0}$.

In the second step, we extend the circuit $F_f^t$ to construct $G_f^t$, addressing the case when the $B$-regiter is in a general state. We follow a standard technique: first, we use $F_f^t$ to compute the state $\ket{f(z)}$; then we apply a $\mathrm{CNOT}$ gate; and finally, we uncompute $\ket{f(z)}$ using $F_f^t$ again. This approach ensures that whenever $f(z) = 1$, the circuit $F_f^t$ is used only when the state in the $B$-register is close to $\ket{0}$, while for $f(z) = 0$, the circuit acts as  identity.


%% file: technical-overview.tex
\subsection*{Organization}

In \cref{sec:boolean-functions}, we construct the circuit $F_f^t$, and in \cref{sec:full-oracle-replacement}, we extend this circuit to $G_f^t$. In \cref{sec:multi-bit-functions}, we discuss how to further extend the circuit to handle functions of the form $f : \{0,1\}^n\to \{0,1\}^m$, where $m \geq 2$. Finally, in \cref{sec:any-algorithm}, we provide a formal proof of Theorem~\ref{thm: Informal Main Theorem}.

%% file: preliminaries.tex
\section{Preliminaries}
\label{sec: prelims}



We denote registers with capital letters, such as $B$, and the Hilbert space corresponding to each register as $\mathcal{H}_B$.
When this is not be clear from the context, we use subscripts to indicate the space that a state belongs to or an operator acts on; for example, $\ket{\psi}_B\in \mathcal{H}_B$. We refer the reader to~\cite{buhrman2002complexity, R03} for background on quantum query complexity.

For $f: \{0,1\}^n\to \{0,1\}^m$, we assume we have access to an underlying oracle $O_f$ that acts as
$$O_f\ket{z}\ket{b} = \ket{z}\ket{b\oplus f(z)}$$
for all $z \in \{0,1\}^n$ and $b \in \{0,1\}^m$, where $\oplus$ denotes the bitwise XOR. 
Typically, we assume that $\ket{z}_Z\ket{b}_B\in \mathcal{H}_Z \otimes \mathcal{H}_B$, where 
$Z$ is referred to as the \textit{input} register and 
$B$ as the \textit{target} register. 
Finally, we use register $S$ as an additional \textit{scratch-pad} register, which is discarded at the end of the computation.


In the \textit{faulty oracle} model,
instead of $O_f$, we are given oracle access to the channel defined as follows 
$$O_{f,p}(\cdot) = \begin{cases}
    O_f(\cdot) & \text{ With probability $(1-p)$}\\
    I(\cdot) & \text{ With probability $p$},
\end{cases}$$
where $p$ is the error-rate and $I$ is the identity map, i.e., $I(\rho) = \rho$ for all matrices $\rho$.
Essentially, with probability $p$,  this oracle acts as the identity on its input state, while with probability $(1-p)$, the oracle acts as $O_f$.



From a formal point of view, we think of $O_{f,p}$ as a random variable that takes one of two (unitary matrix) values: $O_f$ or $I$, where $I$ is the identity {matrix}. This is also how we think of the circuits 
involving $O_{f,p}$ and we formulate our theorems and carry out the proofs over the randomness of this process. This has the  advantage that we can continue treating $O_{f,p}$ as a unitary, even though the faulty oracle entails a random irreversible process. 


We use a geometric approach to prove the correctness of our constructions. For this, we need the following definition.
\begin{definition}[Absolute Angle Difference]
    For $\ket{z,\phi},\ket{z,\psi} \in \mathcal{H}_Z\otimes \mathcal{H}_B$ where $\ket{\psi}$ and $\ket{\phi}$ have real coefficients and $\ket{z}$ is a basis vector from the computational basis, we define the absolute angle difference as
    $$\Phi(\ket{z,\phi},\ket{z,\psi}_B) = |\cos^{-1}\braket{\phi}{\psi}|\in[0,\pi].$$
\end{definition}
We have the following upper bound on the distance between two vectors in terms of absolute angle difference.
\begin{observation}\label{obs:replace-angle-difference}
     For states $\ket{z,\phi},\ket{z,\psi} \in \mathcal{H}_Z\otimes \mathcal{H}_B$ where $\ket{\psi}$ and $\ket{\phi}$ have real coefficients and $\ket{z}$ is a basis vector from the computational basis,
     $$\norm{\ket{z,\phi}-\ket{z,\psi}} \le \Phi(\ket{z,\phi},\ket{z,\psi})$$
     where the norm is the $\ell^2$ norm.
\end{observation}


Restricting ourselves to only a two-dimensional real space spanned by the vectors $\ket{z}\ket{0}$ and $\ket{z}\ket{1}$, we have the following observation about the actions of the faulty oracle in this space.

\begin{observation}\label{obs:oracle-action}
    For any $\ket{\phi}\in\mathbb{R}^2$ and $z\in\{0,1\}^n$, $O_{f,p}$ acts on $\ket{z,\phi}$ like:
    \begin{itemize}
        \item If $f(z) = 0$: Identity
        \item If $f(z) = 1$: A reflection in the vector $\ket{z,+}$ with probability $1-p$ and identity with probability $p$.
    \end{itemize}
\end{observation}


We will repeatedly use the rotation matrix, defined as 
$$R_{\theta} = \begin{bmatrix}
    \cos\theta & \sin\theta\\
    -\sin\theta & \cos\theta
\end{bmatrix}$$
which performs a clockwise rotation by angle $\theta$.


The trace distance between pure states $\ket{\phi}$ and $\ket{\psi}$ is defined as follows
\begin{align*}
    T(\ket{\phi},\ket{\psi}) = \frac{1}{2}\norm{\ketbra{\phi}-\ketbra{\psi})}_{1} = \sqrt{1-|\braket{\phi}{\psi}|^2},
\end{align*}
see~\cite{Watrous_2018} for a proof of the second equality. We have the following lemma. 

\begin{lemma}\label{lem:dist-to-trace-dist}
    For any pure states $\ket{\phi}$ and $\ket{\psi}$, we have
    $$
    T(\ket{\phi},\ket{\psi})
    \le \norm{\ket{\phi}-\ket{\psi}}.$$
\end{lemma}

\begin{proof}
    Observe that
    \begin{align*}
        \norm{\ket{\phi}-\ket{\psi}}^2
        &= 2-2\Re\{\braket{\phi}{\psi}\} \\
        &\geq 2-2|\braket{\phi}{\psi}| \\
        &\geq 2-2|\braket{\phi}{\psi}| - (1 - |\braket{\phi}{\psi}|)^2\\
        &= 1 - |\braket{\phi}{\psi}|^2 \\
        &= T(\ket{\phi},\ket{\psi})^2.
    \end{align*}
\end{proof}
    

%% file: faulty-oracle-noise.tex
\section{Overcoming the Faulty Oracle Noise {for Boolean functions}}\label{sec:boolean-functions}

In this section, we assume that $f$ is a {Boolean} function, that is, $f : \{0,1\}^n \to \{0,1\}$. In~\cref{sec:multi-bit-functions}, we generalize the results of this section to functions of the form $f : \{0,1\}^n \to \{0,1\}^m$ where $m \geq 2$. We will furthermore use $z\in\{0,1\}^{n}$ to denote a bitstring as a possible input to the function $f$.

We now discuss a central component of our algorithm. For the purpose of this algorithm, we assume $p=1/2$. It will soon turn out to give some nice symmetry properties that will be important for the proof of correctness. Note that if $p<1/2$ and $p$ is known, then we can increase the probability of error to $1/2$ by simply not applying the oracle with probability 
$1-\tfrac{1}{2(1-p)}$.


The core part of our algorithm is shown in two figures: \cref{alg:robust-0-oracle} gives the circuit for $F_f^t$, and \cref{alg:robust-full-oracle} gives the circuit for $G_f^t$. 
It easy to see that when $f(z) = 0$ then both $F_f^t$ and $G_f^t$ act as identity. 
Thus, the primary goal will be to show that $((I\otimes R_{\frac{\pi}{2t}})O_{f,\frac{1}{2}})^t$ acts very close to identity on $\ket{z,+}$, with high probability, when $f(z) = 1$.

Intuitively, the circuit $F_f^t$ does the following steps. First, it rotates the target 
register by $-\pi/4$ which results in rotating $\ket{z,0}$ to $\ket{z,+}$. Then it repeatedly reflects the input in $\ket{z,+}$ using the oracle $O_{f,\frac{1}{2}}$ and rotates the state by the angle $\frac{\pi}{2t}$. Finally, it rotates the state by $-\pi/4$. 

Since we are working with a random process, it will be useful to succinctly describe different possible outcomes of this process. 
We will think of the circuits involving the faulty oracle as outputting various \emph{pure} states each with some probability.


Informally, we let $\ket{z,\psi}$ denote the state of the computation just before the current step, where $f(z) = 1$ and $\ket{\psi} = \alpha\ket{0}+\beta\ket{1}$ for some real $\alpha,\beta$. We use the following notation: 
\begin{itemize}
    \item $rf$: Oracle fails while $\ket{z,\psi}$ is on the right of $\ket{z,+}$ 
    \item $rs$: Oracle succeeds while $\ket{z,\psi}$ is on the right of $\ket{z,+}$ 
    \item $lf$: Oracle fails while $\ket{z,\psi}$ is on the left of $\ket{z,+}$ 
    \item $ls$: Oracle succeeds while $\ket{z,\psi}$ is on the left of $\ket{z,+}$ 
\end{itemize}

\begin{figure}
    \centering
        \begin{quantikz}
        \ket{z}_Z&\gate[2]{F_f^t}& \\
        \ket{0}_S&&
        \end{quantikz}=\begin{quantikz}
        \ket{z}_Z&
        &\gate[2]{O_{f,\frac{1}{2}}}\gategroup[2,steps=2,style={rounded corners,dashed,inner sep=6pt, },label style={label
position=above right,anchor=north,xshift=0.40cm,yshift=0cm}]{$\times t$} &  & & & \\
        \ket{0}_S & \gate{R_{-\pi/4}} &  & \gate{R_{\frac{\pi}{2t}}} &  & \gate{R_{-\pi/4}} &
        \end{quantikz}
    \caption{Robust calculation of $O_{f}\ket{z,0}$}
    \label{alg:robust-0-oracle}
\end{figure}
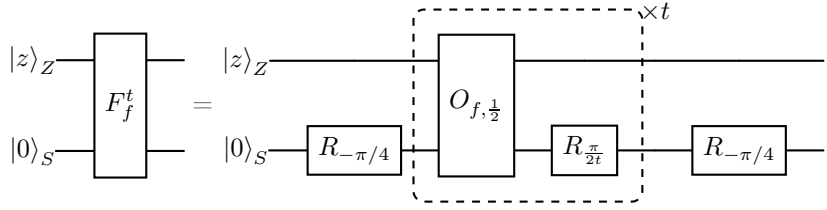

\begin{figure}
    \centering
        \begin{quantikz}
        \ket{z}_Z&\gate[3,nwires=2]{G_f^t}& \\
        \setwiretype{n}&&\\
        \ket{b}_B&&
        \end{quantikz}=\begin{quantikz}
        \ket{z}_Z& & \gate[2]{F_f^t} & & & \gate[2]{F_f^t} & & \\
        \setwiretype{n}&\ket{0}_S&\setwiretype{q} & \ctrl{1} & \gate{X} &  & \gate{X} &\\
        \ket{b}_B&&&\targ{}&&&&
        \end{quantikz}
    \caption{Robust calculation of $O_{f}\ket{z,b}$. 
    }
    \label{alg:robust-full-oracle}
\end{figure}
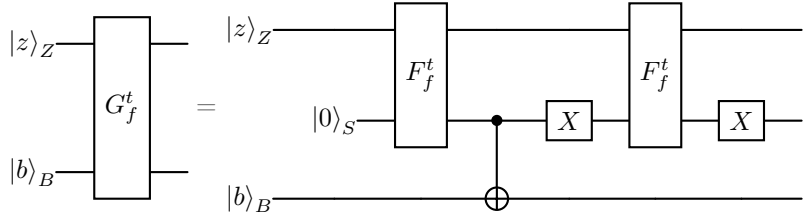

We further define that $\ket{z,+}$ is on the right of $\ket{z,+}$ to simplify proofs and not have to care about equivalence.

\begin{definition}
    We use a sequence $a=a_1,\ldots,a_t$ where $a_i\in \{rf,rs,lf,ls\}$ to describe the success/failure sequence in a fixed execution of the random process described in \cref{alg:robust-0-oracle}. We call such a sequence valid if it can arise from some computation. 
\end{definition}
It should be noted that some sequences in $\{rf,rs,lf,ls\}^t$ do not describe a valid computation. For example, every computation must start with either $rs$ of $rf$ since the initial state is $\ket{z,+}$ which is defined to be on the right.
\begin{definition}
    For a sequence describing a computation $a=a_1,\ldots,a_t$, define $c_i=1$ if $a_i\in \{rs,ls\}$ and $c_i=0$ otherwise. Let $U(a)$ be the unitary:
    $$U(a) = (I\otimes R_{\frac{\pi}{2t}})(O_{f})^{c_t}\ldots(I\otimes R_{\frac{\pi}{2t}})(O_{f})^{c_1}.$$
    This is the unitary that is applied after $t$ repetitions of $(I\otimes R_{\frac{\pi}{2t}})O_{f,\frac{1}{2}}$
\end{definition}



Ignoring the case when the angle difference is $0$, we make a simple observation, which follows easily from considering the action of the different reflections and rotations (see  \cref{obs:oracle-action}). 
\begin{figure}
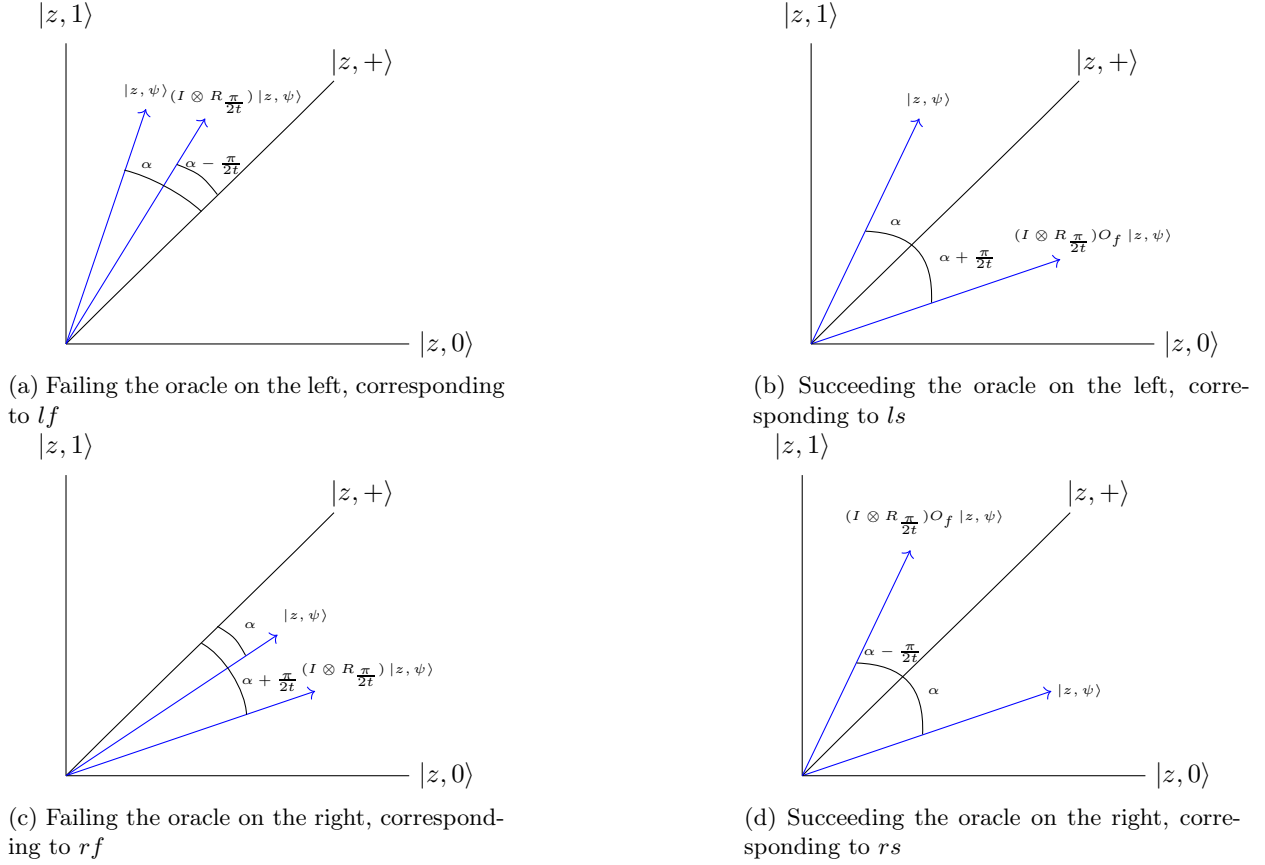

\centering
\begin{subfigure}{0.4\textwidth}
    \ctikzfig{angle-action-lf}
    \caption{Failing the oracle on the left, corresponding to $lf$}
    \label{fig:lf}
\end{subfigure}
\hfill
\begin{subfigure}{0.4\textwidth}
    \ctikzfig{angle-action-ls}
    \caption{Succeeding the oracle on the left, corresponding to $ls$}
    \label{fig:ls}
\end{subfigure}
\hfill
\begin{subfigure}{0.4\textwidth}
    \ctikzfig{angle-action-rf}
    \caption{Failing the oracle on the right, corresponding to $rf$}
    \label{fig:rf}
\end{subfigure}
\hfill
\begin{subfigure}{0.4\textwidth}
    \ctikzfig{angle-action-rs}
    \caption{Succeeding the oracle on the right, corresponding to $rs$}
    \label{fig:rs}
\end{subfigure}
\hfill
\caption{The $4$ possible actions that the oracle $O_{f,p}$ can end up making when $f(z)=1$.}
\label{fig:actions}
\end{figure}

\begin{observation}\label{obs:sequence-actions}
    Let $a=a_1,\ldots,a_t$ be the sequence describing the computation, and let $\ket{z}$ be a computational basis vector where $f(z)=1$. If $\Phi(U(a_1,\ldots,a_{i})\ket{z,+},\ket{z,+})\neq 0$ then
    \begin{itemize}
        \item if $a_i = rf$ then $\Phi(U(a_1,\ldots,a_{i})\ket{z,+},\ket{z,+}) = \Phi(U(a_1,\ldots,a_{i-1})\ket{z,+},\ket{z,+})+\frac{\pi}{2t}$
        \item if $a_i = rs$ then $\Phi(U(a_1,\ldots,a_{i})\ket{z,+},\ket{z,+}) = \Phi(U(a_1,\ldots,a_{i-1})\ket{z,+},\ket{z,+})-\frac{\pi}{2t}$
        \item if $a_i = lf$ then $\Phi(U(a_1,\ldots,a_{i})\ket{z,+},\ket{z,+}) = \Phi(U(a_1,\ldots,a_{i-1})\ket{z,+},\ket{z,+})-\frac{\pi}{2t}$
        \item if $a_i = ls$ then $\Phi(U(a_1,\ldots,a_{i})\ket{z,+},\ket{z,+}) = \Phi(U(a_1,\ldots,a_{i-1})\ket{z,+},\ket{z,+})+\frac{\pi}{2t}$
    \end{itemize}
\end{observation}
An immediate consequence of this is that the angle difference between $\ket{z,+}$ and the result of the computation is always a multiple of $\frac{\pi}{2t}$, which simplifies the analysis.

It is important to note that both when we are on the left side and when we are on the right side there is an outcome where the absolute angle difference decreases by $\frac{\pi}{2t}$ and one where it increases by $\frac{\pi}{2t}$. This is the reason why we force $p=1/2$. Because on the left side, we have the decrease by $\frac{\pi}{2t}$ if the oracle fails, while on the right side, the decrease is when the oracle succeeds. Setting $p=1/2$ 
allows us to
model the difference in angle using a random walk without having to care for orientation. This idea is formalized in the following lemma:

\begin{lemma}
    Let $A=A_1,\ldots,A_t$ where for each $i\in \{1,\ldots,t\}$, $A_i\in\{rf,rs,lf,ls\}$ is the random variable that describes the computation. We define $X_1,\ldots,X_t$ as $X_i = 1$ if and only if $A_i\in \{rf,ls\}$ and $-1$ otherwise. Then $\Pr[X_i=-1] = \Pr[X_i=1] = \frac{1}{2}$ for all $1 \le i \le t$, and all $X_i$ are independent.
\end{lemma}
\begin{proof}
    Regardless of whether we are on the right or left, the probability that the oracle succeeds is the same as the probability  that it fails. So we have $\Pr[A_i=rs] = \Pr[A_i=rf]$ and $\Pr[A_i=ls] = \Pr[A_i=lf]$. Since $A_i\in\{rf,rs,lf,ls\}$ we conclude that $\Pr[A_i\in \{rs,lf\}] = \frac{1}{2}$. Then independence follows from the fact that regardless of previous outcomes, $\Pr[A_i\in \{rs,lf\}] = \frac{1}{2}$.
\end{proof}

The above lemma allows us to use a random walk to bound how much we have moved from our starting point.
\begin{lemma}\label{lem:random-walk-upper-bound}
    Let $a=a_1,\ldots,a_t \in\{rf,rs,lf,ls\}^t$ be a fixed sequence that describes the computation, and let $\ket{z}$ be a standard basis vector where $f(z)=1$. Let $X_1,\ldots,X_t$ be defined such that $X_i = 1$ if and only if $a_i\in \{rf,ls\}$ and $-1$ otherwise. Then 
    $$\Phi(U(a)\ket{z,+},\ket{z,+})= \frac{\pi}{2t}\left|\sum_{i=1}^{t}X_i\right|$$
\end{lemma}
\begin{proof}
    First we observe that for all $r$, $\Phi(U(a_1,\ldots,a_{r})\ket{z,+},\ket{z,+})$ takes values that are integer multiples of $\frac{\pi}{2t}$ by \cref{obs:sequence-actions}.
    
    Now whenever we have that $\Phi(U(a_1,\ldots,a_{r})\ket{z,+},\ket{z,+})\neq 0$, by \cref{obs:sequence-actions} we have that:
    \begin{equation}\label{eq:before-flip}
        \Phi(U(a_1,\ldots,a_{r})\ket{z,+},\ket{z,+}) = \Phi(U(a_1,\ldots,a_{r-1})\ket{z,+},\ket{z,+})+X_l\frac{\pi}{2t}.
    \end{equation}
    We see that this is the same contribution that $X_r$ has to the sum of $\frac{\pi}{2t}\left|\sum_{i=1}^{t}X_i\right|$.
    
    Next, observe that if $\Phi(U(a_1,\ldots,a_{r-1})\ket{z,+},\ket{z,+}) = 0$ then in the next step the angle always increases and  we thus have:
    \begin{equation}\label{eq:after-flip}
        \left|\Phi(U(a_1,\ldots,a_{r})\ket{z,+},\ket{z,+})\right| = \left|\Phi(U(a_1,\ldots,a_{r-1})\ket{z,+},\ket{z,+})+X_{r}\frac{\pi}{2t}\right|.
    \end{equation}
    Therefore by taking the absolute value of \cref{eq:before-flip} and combining it with \cref{eq:after-flip} we get
    $$\Phi(U(a)\ket{z,+},\ket{+}) = \frac{\pi}{2t}\left|\sum_{i=1}^{t}X_i\right|.$$
\end{proof}

From this point onward we are going to focus on using this random walk representation of the absolute difference to construct a bound on the concentration of the distance.
\begin{lemma}\label{lem:prob-bound}
    Let $X_1,\ldots,X_t$ be random variables that take values independently and uniformly at random from the set $\{-1,1\}$. Then with probability at least $1-\delta$:
    $$\left|\sum_{i=1}^{t}X_i\right| < \sqrt{6t\ln (2\delta^{-1})}.$$
\end{lemma}
\begin{proof}
    Define $Y_i = \frac{X_i+1}{2}$. Then $Y_i$ are i.i.d uniform Bernoulli distributed random variables with $E[Y_i] = \frac{1}{2}$. Then
    \begin{align*}
        \Pr[\left|\sum_{i=1}^{t}X_i\right|\ge \sqrt{6 t\ln(2\delta^{-1})}]
        &= \Pr[\left|2\sum_{i=1}^{t}Y_i - t\right|\ge \sqrt{6 t\ln(2\delta^{-1})}]\\
        &= \Pr[\left|\sum_{i=1}^{t}Y_i - \frac{t}{2}\right|\ge \sqrt{\frac{6\ln(2\delta^{-1})}{t} }\cdot \frac{1}{2}t ]\\
        &\le 2\exp(-\frac{6\ln (2\delta^{-1})}{t}\cdot \frac{1}{6} t)\\
        &= \delta
    \end{align*}
    where we have used Chernoff bound in the second inequality.
\end{proof}

We can now prove the main statement of this section, namely that when $t$ gets large enough, $((I\otimes R_{\frac{\pi}{2t}})O_{f,\frac{1}{2}})^t$ acts almost with identity on $\ket{z,+}$. 

\begin{lemma}\label{lem:reflection-rotation-identity}
    Let $t \ge \frac{\frac{3}{2}\pi^2\ln(2\delta^{-1})}{\gamma^2}$ and $\delta \le \frac{1}{5}$. Then with probability at least $1-\delta$, for any state $\ket{z,+}$ where $\ket{z}$ is a basis vector in the computational basis with $f(z)=1$
    $$\Phi(((I\otimes R_{\frac{\pi}{2t}})O_{f,\frac{1}{2}})^t\ket{z,+},\ket{z,+}) \le \gamma.$$
\end{lemma}
\begin{proof}
    Let $A=A_1,\ldots,A_t$ where for each $i\in \{1,\ldots,t\}$, $A_i\in\{rf,rs,lf,ls\}$ be a sequence of random variables that describe the computation. Let $X_1,\ldots,X_t$ be defined such that $X_i = 1$ if and only if $A_i\in \{rf,ls\}$ and $-1$ otherwise. By \cref{lem:random-walk-upper-bound} we have
    $$\Phi(((I\otimes R_{\frac{\pi}{2t}})O_{f,\frac{1}{2}})^t\ket{z,+},\ket{z,+})\le \frac{\pi}{2t}\left|\sum_{i=r}^{t}X_i\right|.$$
    Then by \cref{lem:prob-bound} with probability at least $1-\delta$:
    \begin{align*}
        \frac{\pi}{2t}\left|\sum_{i=r}^{t}X_i\right| &\le \frac{\pi}{2t}\sqrt{6t\ln (2\delta^{-1})}\\
        &\le \frac{\pi\sqrt{\frac{3}{2}\ln (2\delta^{-1})}}{\sqrt{t}}\\
        &\le \gamma,
    \end{align*}
    the last inequality follows from $t \ge \frac{6\pi^2\ln(4^{-1}\delta^{-1})}{\gamma^2}$.
\end{proof}

%% file: full-oracle-replacement.tex
\subsection{Generalizing to all Inputs}\label{sec:full-oracle-replacement}


In this section, we will continue the work of proving that $G_f^t$ acts roughly equivalent to $O_f$. Most of the work in this section is essentially to show that we can go from acting correctly on a subset of the basis vectors, as was shown with \cref{lem:reflection-rotation-identity}, to a general state. The first step is checking for product states when $f(z)=0$. This is very much the easy case since both the faulty and non-faulty oracles act the same way.

\begin{lemma}\label{lem:z-0-gives-identity}
    For any state $\ket{z,\psi}$ where $\ket{z}$ is a basis vector in the computational basis and $f(z)=0$,
    $$F_f^t \ket{z,\psi} = O_f \ket{z,\psi}$$
\end{lemma}
\begin{proof}
    If $f(z)=0$ then $O_f\ket{z,\psi} = \ket{z,\psi}$. Furthermore from the definition of $O_{f,\frac{1}{2}}$, we have that when $f(z)=0$ then it always holds that $O_{f,\frac{1}{2}}(\ket{z,\psi}) = \ket{z,\psi}$. This means that 
    $$F_f^t(\ket{z,\psi}) = (I\otimes R_{-{\pi/4}}) (I\otimes R_{\frac{\pi}{2t}})^t (I\otimes R_{-{\pi/4}})\ket{z,\psi} = \ket{z,\psi}$$
\end{proof}

From \cref{lem:reflection-rotation-identity}, we can add the final couple of gates 
and look at how this affects a basis state in the computational basis when $f(z)=1$.
\begin{lemma}\label{lem:robust-0-oracle-concentration}
    Let $t \ge \frac{\frac{3}{2}\pi^2\ln(2\delta^{-1})}{\gamma^2}$ and let $F_f^t$ be defined as in \cref{alg:robust-0-oracle}. Then with probability at least $1-\delta$ for any basis vector $\ket{z,0}$ from the computational basis such that $f(z)=1$,
    $$\norm{(F_f^t-O_f)\ket{z,0}} \le \gamma$$
\end{lemma}
\begin{proof}
    Let $A=A_1,\ldots,A_t$ be a set of random variables that describe the computation. Then $U(A)$ is the actual unitary that is applied by $((I\otimes R_{\frac{\pi}{2t}})O_{f,\frac{1}{2}})^t$. 
    From \cref{alg:robust-0-oracle} we have $F_f^t = (I\otimes R_{-\pi/4})U(A)(I\otimes R_{-\pi/4})$.

    We are going to show that $F_f^t$ acts similarly to $O_f$ on all basis vectors from the computational basis $\ket{z,0}$. For this we split the basis vectors into two cases, $f(z)=0$ and $f(z) = 1$. The $f(z)=0$ follows from \cref{lem:z-0-gives-identity}. Therefore we are going to assume that $f(z)=1$.

    By the triangle inequality, \cref{obs:replace-angle-difference} and \cref{lem:reflection-rotation-identity} we have that with probability at least $1-\delta$: 
    \begin{align*}
        \norm{(F_f^t-O_f)\ket{z,0}} &= \norm{((I\otimes R_{-\pi/4})U(A)(I\otimes R_{-\pi/4})\ket{z,0}-\ket{z,1}}\\
        &= \norm{((I\otimes R_{-\pi/4})U(A)\ket{z,+}-\ket{z,1}}\\
        &\le \norm{((I\otimes R_{-\pi/4})\ket{z,+}-\ket{z,1}}+\norm{((I\otimes R_{-\pi/4})\ket{z,+}-(I\otimes R_{-\pi/4})U(A)\ket{z,+}}\\
        &\le \gamma
    \end{align*}
    completing the proof.
\end{proof}

Next, we show that applying our random process on a basis vector $\ket{z}_Z\ket{0}_S$ does not affect the input register $Z$, but only the output register when $f(z)=1$, in this case $S$ (see~\cref{alg:robust-0-oracle}). This will imply that all of our error is going to be concentrated in the $S$ register, which makes it easier to bound.

\begin{lemma}\label{lem:modify-only-output}
    Let $F_f^t$ be defined as in \cref{alg:robust-0-oracle}, and let $\ket{z,0}$ be a basis vector from the computational basis such that $f(z)=1$. Then there exists a state $\ket{\psi}$ (only dependent on the randomness of $F_f^t$) such that
    $$F_f^t\ket{z,0} = \ket{z,\psi}$$
\end{lemma}
\begin{proof}
    Fix $a$ as the sequence of computation. Then $F_f^t= (I\otimes R_{-\pi/4})U(a)(I\otimes R_{-\pi/4})$. Next observe that since $f(z)=1$, $O_f\ket{z,0} = (I\otimes X)\ket{z,0}$. If make this substitution for every $O_f$ unitary in $U(a)$, then we can write $F_f^t\ket{z,0} = (I\otimes V(a))\ket{z,0}$ for some unitary $V(a)$. Defining $\ket{\psi} = V(a)\ket{0}$ completes the proof.
\end{proof}

So far we have only looked at basis vectors, but we will now continue with arbitrary vectors and show that when we have $\ket{0}$ in the output state we still achieve the intended outcome for $F_f^t$.
\begin{lemma}\label{lem:robust-0-oracle-concentration-all}
    Let $t \ge \frac{\frac{3}{2}\pi^2\ln(2\delta^{-1})}{\gamma^2}$ and let $F_f^t$ be defined as in \cref{alg:robust-0-oracle}. Then with probability at least $1-\delta$ for any state $\ket{\phi,0}$,
    $$\norm{(F_f^t-O_f)\ket{\phi,0}} \le \gamma.$$
\end{lemma}
\begin{proof}
    Let $\ket{\phi} = \sum_{z}\alpha_z\ket{z}$ for basis vectors $\ket{z}$ in the computational basis. We have
    \begin{align*}
        \norm{(F_f^t-O_f)\ket{\phi,0}} &= \norm{(F_f^t-O_f)\sum_{z}\alpha_z\ket{z,0}}\\
        &= \norm{\sum_{z : f(z)=0 }\alpha_z(F_f^t-O_f)\ket{z,0} + \sum_{z : f(z)=1 }\alpha_z(F_f^t-O_f)\ket{z,0}}.\\
    \end{align*}
    We have by \cref{lem:z-0-gives-identity} that when $f(z)=0$, $F_f^t\ket{z,0}=O_f\ket{z,0}$, enabling us to remove the first term of the last equation. On the other hand, when $f(z) = 1$, then $O_f\ket{z,0}=\ket{z,1}$. Thus
    \begin{align*}
         \norm{\sum_{z : f(z)=0 }\alpha_z(F_f^t-O_f)\ket{z,0} + \sum_{z : f(z)=1 }\alpha_z(F_f^t-O_f)\ket{z,0}}&=\norm{\sum_{z : f(z)=1 }\alpha_z (F_f^t\ket{z,0}-\ket{z,1})}.
    \end{align*}
    We can now without loss of generality assume that there exists a value $z_0$ such that $f(z_0)=1$, since otherwise the norm is $0$ and the statement holds.
    
    Next, by \cref{lem:modify-only-output}, when $f(z)=1$ we have $F_f^t\ket{z,0} = \ket{z,\psi}$ for some $\ket{\psi}$, where $\ket{\psi}$ is only dependent on the randomness of $F_f^t$. As this randomness is the same regardless of the value of $z$ we get, we have
    \begin{align*}
        \norm{\sum_{z : f(z)=1 }\alpha_z (F_f^t\ket{z,0}-\ket{z,1})} &= \norm{\sum_{z : f(z)=1 }\alpha_z (\ket{z,\psi}-\ket{z,1})}\\
        &= \norm{\left(\sum_{z : f(z)=1 }\alpha_z\ket{z}\right)\otimes(\ket{\psi}-\ket{1})}\\
        &\le \norm{\ket{z_0}\otimes(\ket{\psi}-\ket{1})}\\
        &= \norm{(F_f^t-O_f)\ket{z_0,0}},
    \end{align*}
    where the inequality follows since $\sum_{z : f(z)=1 }\alpha_z\ket{0}$ has norm at most $1$, and the last equality again using \cref{lem:modify-only-output}. Finally, applying \cref{lem:robust-0-oracle-concentration} completes the proof.
\end{proof}

So far we have worked on showing that $F_f^t$ acts correctly. Crucially,  it only acts correctly when the initial state of the output register is $\ket{0}$. This is however not enough to get it to be completely identical to $O_f$. To get around this, we add a scratchpad register $S$ which we have full control over, and $S$ initially starts in state $\ket{0}$. 
By doing this, we can construct the circuit $G_f^t$ that approximately implements $O_f$.
\begin{theorem}\label{thm:bound-full-state}
    Let $t \ge \frac{6\pi^2\ln(4\delta^{-1})}{\gamma^2}$. Let $G_f^t$ be defined as in \cref{alg:robust-full-oracle}. Then with probability at least $1-\delta$ for any state $\ket{\phi}_{ZB}\ket{0}_S$,
    $$\norm{(G_f^t-((O_f)_{ZB}\otimes I_S))\ket{\phi}_{ZB}\ket{0}_S} \le \gamma.$$
\end{theorem}
\begin{proof}
    From \cref{alg:robust-full-oracle} we have 
    $$G_f^t = (((I_Z \otimes X_S)(F_f^t)_{ZS}(I_Z \otimes X_S))\otimes I_B) (I_Z\otimes CNOT_{SB}) ((F_f^t)_{ZS}\otimes I_B).$$
    We split this into three parts:
    $$K = ((F_f^t)_{ZS}\otimes I_B)$$
    $$L=(I_Z\otimes CNOT_{SB})$$
    $$M =(((F_f^t)_{ZS}(I_Z \otimes X_S))\otimes I_B)$$
    $$R = (I_{ZB} \otimes X_S)$$
    so that $G_f^t = R M L K$. Let $\ket{\phi} = \sum_{z}\sum_{b}\alpha_{zb}\ket{z}_Z\ket{b}_B$. 
    Using $\delta'=\delta/2$ and $\gamma'=\gamma/2$ we have by \cref{lem:robust-0-oracle-concentration-all} with probability $1-\delta/2$:
    \begin{equation}
        \norm{(K - (O_f)_{ZS}\otimes I_B)\ket{\psi}_{ZB}\ket{0}_S} \le \gamma/2.
    \end{equation}
    Next we expand $((O_f)_{ZS}\otimes I_B)\ket{\psi}_{ZB}\ket{0}_S$ as follows
    \begin{align}
        ((O_f)_{ZS}\otimes I_B)\ket{\psi}_{ZB}\ket{0}_S &= \sum_z \alpha_z((O_f)_{ZS}\otimes I_B)\ket{z,b}_{ZB}\ket{0}_S\nonumber\\
        &= \sum_z \alpha_z\ket{z,b}_{ZB}\ket{f(z)}_S.\label{eq:first-application-close}
    \end{align}
    Now applying $L$ we get 
    \begin{equation*}
        L\sum_z \alpha_z\ket{z,b}_{ZB}\ket{f(z)}_S = \sum_z \alpha_z\ket{z,b\oplus f(z)}_{ZB}\ket{f(z)}_S.
    \end{equation*}
    Finally, we are going to apply 
    $$M-(((O_f)_{ZS}(I_Z \otimes X_S))\otimes I_B) = (((F_f^t-O_f)_{ZS}(I_Z \otimes X_S))\otimes I_B)$$
    and take the norm which gives
    \begin{align}
        &\norm{(((F_f^t-O_f)_{ZS}(I_Z \otimes X_S))\otimes I_B)\sum_z \alpha_z\ket{z,b\oplus f(z)}_{ZB}\ket{f(z)}_S} \nonumber\\
        =&\norm{(((F_f^t-O_f)_{ZS})\otimes I_B)\sum_z \alpha_z\ket{z,b\oplus f(z)}_{ZB}\ket{f(z)\oplus 1}_S}\nonumber\\
        =&\norm{\left(\sum_{z:f(z)=0} \alpha_z((F_f^t-O_f)_{ZS}\otimes I_B)(\ket{z,b}_{ZB}\ket{1}_S) + \sum_{z:f(z)=1} \alpha_z((F_f^t-O_f)_{ZS}\otimes I_B)(\ket{z,b\oplus 1}_{ZB}\ket{0}_S)\right)}.\nonumber
    \end{align}
    For the first term we have that $((F_f^t-O_f)_{ZS}\otimes I_B)(\ket{z,b\oplus 1}_{ZB}\ket{0}_S = 0$, when $f(z)=0$ by \cref{lem:z-0-gives-identity}. This simplifies the calculation to
    \begin{align}
        =&\norm{\left(\sum_{z:f(z)=1} \alpha_z((F_f^t-O_f)_{ZS}\otimes I_B)(\ket{z,b\oplus 1}_{ZB}\ket{0}_S)\right)}\nonumber\\
        =&\norm{\left(\sum_{z:f(z)=1} \alpha_z((F_f^t-O_f)_{ZS})(\ket{z}_{Z}\ket{0}_S)\right)}\nonumber\\
        =&\norm{(F_f^t-O_f)_{ZS}\left(\sum_{z:f(z)=1} \alpha_z(\ket{z}_{Z}\ket{0}_S)\right)}.\nonumber
    \end{align}
    As $\sum_{z:f(z)=1} \alpha_z(\ket{z}_{Z}\ket{0}_S)$ has norm at most $1$ we then get by \cref{lem:robust-0-oracle-concentration-all} using $\delta'=\delta/2$ and $\gamma'=\gamma/2$ that with probability $1-\delta/2$:
    \begin{equation}\label{eq:second-application-close}
        \norm{(F_f^t-O_f)_{ZS}\left(\sum_{z:f(z)=1} \alpha_z(\ket{z}_{Z}\ket{0}_S)\right)}\le \gamma/2.
    \end{equation}
    Finally, it is straightforward to see that 
    \begin{equation}\label{eq:full-is-same}
        R(((O_f)_{ZS}(I_Z \otimes X_S))\otimes I_B)L(O_f)_{ZS}\otimes I_B = (O_f)_{ZB}\otimes I_S
    \end{equation}
    by studying the actions of basis vectors, since this is \cref{alg:robust-full-oracle} but with the $F_f^t$ replaced with $O_f$.
    
    With probability of $1-\delta$ from a union bound, we can now use the triangle inequality as well as \cref{eq:first-application-close}, \cref{eq:second-application-close} and \cref{eq:full-is-same} to bound the norm as:
    \begin{align*}
        &\norm{(G_f^t-((O_f^t)_{ZB}\otimes I_S))\ket{\phi,0}_{ZBS}}\\
        \le& \norm{RML(K - (O_f)_{ZS}\otimes I_B)\ket{\phi,0}_{ZBS}} \\
        &+ \norm{R(M-(((O_f)_{ZS}(I_Z \otimes X_S))\otimes I_B))L(O_f)_{ZS}\otimes I_B)\ket{\phi,0}_{ZBS}}\\
        &+ \norm{(R(((O_f)_{ZS}(I_Z \otimes X_S))\otimes I_B)L(O_f)_{ZS}\otimes I_B-(O_f)_{ZB}\otimes I_S)\ket{\phi,0}_{ZBS}}\\
        =& \gamma/2 + \gamma/2 + 0 = \gamma.
    \end{align*}
\end{proof}

%% file: extending-to-multiple-bits.tex
\section{Beyond Boolean functions}\label{sec:multi-bit-functions}
We will now look at how the description from the previous section can be extended to oracles that do not output just a single bit.
We are going to look at a function of the form $f: \{0,1\}^n\to \{0,1\}^m$, with a corresponding oracle $O_f$ (see Section~\ref{sec: prelims}). We are still given access to the oracle $O_{f,p}$ and want to construct an oracle that is very close to $O_f$.

In this setting, we are going to continue to denote the input register as $Z$ and the output register as $\mathcal{H}_B = \bigotimes_{i=1}^{m}\mathcal{H}_{B_i}$ which is composed of $m$ single qubits, that is $\mathcal{H}_{B_i} \cong \mathbb{C}^2$ for all $i \in \{1, \dots, m\}$.

For this, we are going to observe that $f$ corresponds to $m$ boolean functions, each computing a single bit of the output. The idea is therefore going to be that we essentially perform the circuit from \cref{alg:robust-full-oracle} $m$ times in parallel, such that we error-correct each of the $m$ boolean functions. What is interesting is that we do not particularly have to touch the probability of success, since for the $m$ parallel repetitions they are going to be perfectly correlated with respect to their error.

We are going to define $G_f^t$ for non-Boolean functions simply by replacing each of the single qubit gates with gates that act on each of the qubits instead, which can be seen in \cref{alg:robust-0-oracle-multi} and \cref{alg:robust-full-oracle-multi}.

\begin{figure}
    \centering
        \begin{quantikz}
        \ket{z}_Z&\gate[2]{F_f^t}& \\
        \ket{0^m}_S&&
        \end{quantikz}=\begin{quantikz}
        \ket{z}_A&
        &\gate[2]{O_{f,\frac{1}{2}}}\gategroup[2,steps=2,style={rounded corners,dashed,inner sep=6pt, },label style={label
position=above right,anchor=north,xshift=0.9cm,yshift=0cm}]{$\times t$} &  & & & \\
        \ket{0^m}_S & \gate{\bigotimes_{i=1}^m R_{-\pi/4}} &  & \gate{\bigotimes_{i=1}^m R_{\frac{\pi}{2t}}} &  & \gate{\bigotimes_{i=1}^m R_{-\pi/4}} &
        \end{quantikz}
    \caption{Robust calculation of $O_{f}\ket{z,0}$ for non-Boolean $f$.}
    \label{alg:robust-0-oracle-multi}
\end{figure}
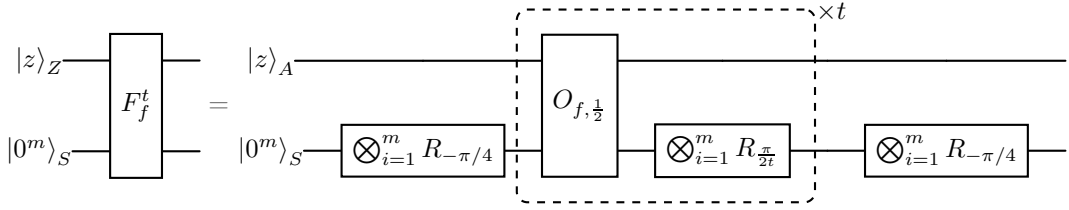

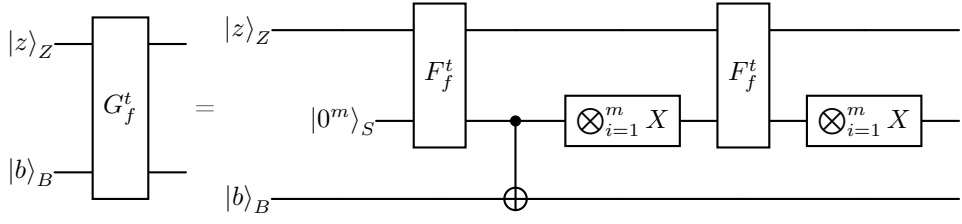
\begin{figure}
    \centering
        \begin{quantikz}
        \ket{z}_Z&\gate[3,nwires=2]{G_f^t}& \\
        \setwiretype{n}&&\\
        \ket{b}_B&&
        \end{quantikz}=\begin{quantikz}
        \ket{z}_Z& & \gate[2]{F_f^t} & & & \gate[2]{F_f^t} & & \\
        \setwiretype{n}&\ket{0^m}_S&\setwiretype{q} & \ctrl{1} & \gate{\bigotimes_{i=1}^m X} &  & \gate{\bigotimes_{i=1}^m X} &\\
        \ket{b}_B&&&\targ{}&&&&
        \end{quantikz}
    \caption{Robust calculation of $O_{f}\ket{z,b}$ for non-Boolean $f$. 
    }
    \label{alg:robust-full-oracle-multi}
\end{figure}

\begin{theorem}\label{thm:bound-non-bool-function}
     Let $t \ge \frac{6\pi^2 m^2\ln(4\delta^{-1})}{\gamma^2}$. Let $G_f^t$ be defined as in \cref{alg:robust-full-oracle-multi}. Then with probability at least $1-\delta$ for any state $\ket{\phi}_{ZB}\ket{0}_S$,
    $$\norm{(G_f^t-((O_f^t)_{ZB}\otimes I_S))\ket{\phi}_{ZB}\ket{0^m}_S} \le \gamma.$$
\end{theorem}
\begin{proof}
    For the function $f$ we define $f_i, i\in \{1,\ldots,m\}$, such that $f_i(z)$ is the $i$'th bit of $f(z)$. Let $I_{S_{-i}}=\bigotimes_{j=1,j\neq i}^{m} I_{S_j}$ and similar for $I_{B_{-i}}$. We see that $G_f^t = \prod_{i=1}^m (G_{f_i}^t)_{ZS_iB_i}\otimes I_{S_{-i}B_{-i}}$, and that $O_f^t = \prod_{i=1}^m (O_{f_i})_{ZB_i} \otimes  I_{B_{-i}}$, where it should be noted that the $G^t_{f_i}$ are defined by \cref{alg:robust-full-oracle} and are perfectly correlated with each other since they share the same underlying $O_{f,\frac{1}{2}}$. So from a couple of triangle inequalities, we get that 
    \begin{align*}
        &\norm{G_f^t-((O_f^t)_{ZB}\otimes I_S))\ket{\phi}_{ZB}\ket{0^m}_S}\\
        \le & \max_{\ket{\psi}_{ZB}}\norm{G_f^t-((O_f^t)_{ZB}\otimes I_S))\ket{\psi}_{ZB}\ket{0^m}_S}\\
        =& \max_{\ket{\psi}_{ZB}}\norm{\left(\left(\prod_{i=1}^m (G_{f_i}^t)_{ZS_iB_i}\otimes I_{S_{-i}B_{-i}}\right)-\left(\left(\prod_{i=1}^m (O_{f_i})_{ZB_i} \otimes  I_{B_{-i}}\right)\otimes I_S\right)\right)\ket{\psi}_{ZB}\ket{0^m}_S} \\
        \le & \sum_{i=1}^m\max_{\ket{\psi}_{ZB_i}}\norm{\left(\left(G_{f_i}^t\right)_{ZS_iB_i}-\left(\left( O_{f_i}\right)_{ZB_i}\otimes I_{S_i}\right)\right)\ket{\psi}_{ZBi}\ket{0}_{S_i}}.
    \end{align*}
    Now by \cref{thm:bound-full-state}, with $\gamma'=\gamma/m$, we can bound each term in the sum. Furthermore, since each $G_{f_i}^t$ uses the same underlying oracle, we get that they are perfectly correlated in when they either succeed or fail. This means that with probability $1-\delta$:
    $$\sum_{i=1}^m\max_{\ket{\psi}_{ZB_i}}\norm{\left(\left(G_{f_i}^t\right)_{ZS_iB_i}-\left(\left( O_{f_i}\right)_{ZB_i}\otimes I_S\right)\right)\ket{\psi}_{ZB_i}\ket{0^m}_S} \le \sum_{i=1}^m \gamma/m = \gamma.$$
\end{proof}

%% file: arbitrary-algorithms.tex
\section{Error-correcting arbitrary algorithms: Proof of Theorem~\ref{thm: Informal Main Theorem}}
\label{sec:any-algorithm}
We now have a way to efficiently construct a robust oracle, both for Boolean functions due to \cref{thm:bound-full-state} and for non-Boolean functions due to \cref{thm:bound-non-bool-function}. From this, it is going to be straightforward to convert any algorithm into a version that is robust with respect to faulty oracles $O_{f,p}$ where $p\le \frac{1}{2}$ is known. We are simply going to replace each occurrence of the fault-free oracle in the original algorithm with our robust version of the oracle.

We are going to model an algorithm as a series of oracle calls separated by unitaries and we let the algorithm act on three different registers, $Z$, $B$ and $T$. Furthermore, we allow the oracle to act only on the registers $Z$ and $B$, such that $T$ corresponds to the register where the algorithm can both store information between oracle calls and perform additional computations.

Letting $q$ be the number of fault-free oracle calls, we can write the unitary that our arbitrary algorithm implements as: 
$$V = U_q (O_f\otimes I_T) U_{q-1} (O_f\otimes I_T) \ldots U_1 (O_f\otimes I_T)U_0.$$
For the robust algorithm, we furthermore give it access to the space $S$. This means that we can define the unitary where we have replaced each oracle call with a robust version:

$$V' = (U_q\otimes I_S) (G_f^t\otimes I_T) (U_{q-1}\otimes I_S)  (G_f^t\otimes I_T) \ldots (U_1\otimes I_S)  (G_f^t\otimes I_T) U_0.$$

\begin{theorem}\label{thm:bound-algorithm}
     Let $t \ge \frac{6\pi^2 q^2 m^2\ln(4q\delta^{-1})}{\gamma^2}$. Let $G_f^t$ be defined as in \cref{alg:robust-full-oracle-multi}. Let 
     $$V = U_q (O_f\otimes I_{T}) U_{q-1} (O_f\otimes I_{T}) \ldots U_1 (O_f\otimes I_{T}),$$
     $$V' = (U_q\otimes I_S)  (G_f^t\otimes I_T) \ldots (U_1\otimes I_S) (G_f^t\otimes I_T).$$
     Then with probability at least $1-\delta$ for any state $\ket{\phi}_{ZBT}$,
    $$\norm{((V\otimes I_S)-V')\ket{\phi}_{ZBT}\otimes\ket{0^m}_S} \le \gamma.$$
\end{theorem}
\begin{proof}
    The proof simply follows from applying the triangle inequality several times. Define 
    $$V_i = (U_q (O_f\otimes I_{ST}) U_{q-1} \ldots (O_f\otimes I_T)U_i(G_f^t\otimes I_{ST})\ldots U_1 (G_f^t\otimes I_{ST})U_0.$$
    Then we have that $V_q = V'$ and $\tr_S V_0 = V$, where $\tr_S V_0$ denotes the partial trace with respect to $S$. For any $i$ we further have that with probability $1-\frac{\delta}{q}$ by \cref{thm:bound-non-bool-function}:
    \begin{equation*}\label{eq:telescope-alg}
        \norm{(V_{i}-V_{i-1})\ket{\psi}_{ZBT}\ket{0^m}_S} \le \max_{\ket{\psi}_{ZBT}} \norm{((O_f\otimes I_{ST})-(G_f^t\otimes I_T))\ket{\psi}_{ZBT}\ket{0^m}_S}\le \frac{\gamma}{q}.
    \end{equation*}
    Due to the fact that $V_{i-1}$ an $V_i$ only differ in the operator between $U_{i-1}$ and $U_{i}$. We can now use this equation $q$ times to get by a union bound that with probability at least $1-\delta$:
    \begin{align*}
        \norm{(V_q-V_0)\ket{\psi}_{ZBT}\ket{0^m}_S} &\le \sum_{i=1}^{q} \max_{\ket{\psi}_{ZBT}}\norm{(V_i-V_{i-1})\ket{\psi}_{ZBT}\ket{0^m}_S}\\
        &\le \sum_{i=1}^{q}\frac{\gamma}{q} = \gamma
    \end{align*}
    using the triangle inequality in the first inequality.
\end{proof}

From this, we can conclude that we can efficiently and robustly simulate the fault-free query algorithm with respect to the faulty oracle. The robust algorithm then discards the register $S$ and performs the same measurement as the fault-free algorithm. The correctness of the robustified algorithm can be seen via the following corollary, for small enough constants $\gamma$ and $\delta$.


\begin{corollary}
    Let $t \ge \frac{6\pi^2 q^2 m^2\ln(4q\delta^{-1})}{\gamma^2}$. Let $G_f^t$ be the quantum channel as defined in \cref{alg:robust-full-oracle-multi}. Let 
    $$V = U_q (O_f\otimes I_{T}) U_{q-1} (O_f\otimes I_{T}) \ldots U_1 (O_f\otimes I_{T})$$
    $$V' = (U_q\otimes I_S)  (G_f^t\otimes I_T) \ldots (U_1\otimes I_S) (G_f^t\otimes I_T).$$
    Then for any state $\ket{\phi}_{ZBT}$,
    $$T(V\ketbra{\phi}V^{*},\tr_{S}V'(\ketbra{\phi}_{ZBT}\otimes \ketbra{0^m}_S)V'^{*}) \le \gamma+\delta.$$
\end{corollary}
\begin{proof}
    We let $\rho = V'(\ketbra{\phi}\otimes \ketbra{0^m})V'^{*}$. From this $\rho$ can be written as a sum $\rho = \alpha \rho_1 + \beta\rho_2$ with $\alpha+\beta = 1$ where $\rho_1$ is a mixed state 
    that satisfy $\norm{((V\otimes I_S)-V')\ket{\phi}_{ZBT}\otimes\ket{0^m}_S} \le \gamma$, while $\rho_2$ consists of states that does not satisfy this. Then by \cref{thm:bound-algorithm} we have $\beta \le \delta$. Furthermore by \cref{lem:dist-to-trace-dist} and \cref{thm:bound-algorithm} we have $T(\rho_1,(V\ketbra{\phi}V^{*})\otimes \ketbra{0^m})\le \gamma$. From this we get that
    $$T((V\ketbra{\phi}V^{*})\otimes \ketbra{0^m},\rho) \le \alpha T(V\ketbra{\phi}V^{*}\otimes \ketbra{0^m},\rho_1) +  \beta T(V\ketbra{\phi}V^{*}\otimes \ketbra{0^m},\rho_2) \le \delta + \gamma.$$
    Then the statement follows from the fact that applying the partial trace can only decrease the partial trace.
\end{proof}

\section{Conclusion}
\label{sec: conclusion}
In this paper, we demonstrate that under the faulty oracle model, all quantum query algorithms can be made robust to noise, implying that a quantum advantage is preserved for any quantum algorithm that has a near super-cubic advantage over its classical counterpart. Specifically, if a super-polynomial quantum advantage exists for a given problem, such an advantage is retained even under the noisy oracle. This robustness is achieved by replacing each oracle call of a $q$-query quantum algorithm in the fault-free case with a subroutine that approximates each fault-free oracle call to within a distance of $O(1/q)$. This procedure incurs a multiplicative overhead of $O(q^2 \log q)$ queries to the faulty oracle, and its correctness is analyzed using a one-dimensional random walk.

We conclude by highlighting some open problems. Our robustification subroutine, as described above, replaces each fault-free oracle call with a subroutine that queries the faulty oracle $O(q^2 \log q)$ times. Is this overhead optimal, or can it be improved? A lower bound of $\Omega(q)$ queries to the faulty oracle follows from~\cite{RS08}.
Our algorithm works under the assumption that the error-rate $p \leq 1/2$ (see Section~\ref{sec: prelims} for the definition of error-rate) of the  faulty oracle is known. 
We believe that the algorithm can be adapted to handle the case where 
$p$ is unknown, and we mention this as an open question for future work. 
Finally, an important area for further investigation is understanding the impact of other, more “natural” noise models (e.g., decoherence and depolarizing noise, see~\cite{childs2023streaming, R03}) on quantum query algorithms.